\documentclass[10pt, twocolumn, conference]{IEEEtran}    % Use this line for a4

\IEEEoverridecommandlockouts                              % This command is only
\title{MIMO Two-way Relay Channel:\\ Diversity-Multiplexing Tradeoff Analysis}

\author{%
  \authorblockN{Deniz G\"{u}nd\"{u}z\authorrefmark{1}\authorrefmark{2},
    Andrea Goldsmith\authorrefmark{2}, H. Vincent Poor\authorrefmark{1},
  }\\
  \authorblockA{%
    \authorrefmark{1}Department of Electrical Engineering, Princeton University, Princeton, NJ.\\
  }
  \authorblockA{%
    \authorrefmark{2}Department of Electrical Engineering, Stanford University, Stanford, CA.\\
  }
  Email: dgunduz@princeton.edu, andrea@wsl.stanford.edu, poor@princeton.edu
  \thanks{This research was supported by the National Science Foundation under Grants ANI-03-38807 and CNS-06-25637, the DARPA ITMANET program under Grant 1105741-1-TFIND, and the U.S. Army Research Office under MURI award W911NF-05-1-0246.}
}
\date{November, 2007}
\usepackage{amsfonts}
\usepackage{amssymb}
\usepackage{amsmath}
\usepackage{amscd}
\usepackage{psfrag, graphicx}

\usepackage{dsfont}

\newtheorem{thm}{Theorem}[section]
\newtheorem{cor}[thm]{Corollary}

\newtheorem{prop}[thm]{Proposition}

\begin{document}
%\fontsize{12}{14}\selectfont
\maketitle
\thispagestyle{empty}
\pagestyle{empty}

%%%%%%%%%%%%%%%%%%%%%%%%%%%%%%%%%%%%%%%%%%%%%%%%%%%%%%%%%%%%%%%%%%%%%%%%%%%%%%%%

\begin{abstract}
A multi-hop two-way relay channel is considered in which all the terminals are equipped with multiple antennas. Assuming independent quasi-static Rayleigh fading channels and channel state information available at the receivers, we characterize the optimal diversity-multiplexing gain tradeoff (DMT) curve for a full-duplex relay terminal. It is shown that the optimal DMT can be achieved by a compress-and-forward type relaying strategy in which the relay quantizes its received signal and transmits the corresponding channel codeword. It is noteworthy that, with this transmission protocol, the two transmissions in opposite directions can achieve their respective single user optimal DMT performances simultaneously, despite the interference they cause to each other. Motivated by the optimality of this scheme in the case of the two-way relay channel, a novel dynamic compress-and-forward (DCF) protocol is proposed for the one-way multi-hop MIMO relay channel for a half-duplex relay terminal, and this scheme is shown to achieve the optimal DMT performance.
\end{abstract}

%%%%%%%%%%%%%%%%%%%%%%
\section{Introduction}
%%%%%%%%%%%%%%%%%%%%%%

Relays have found applications in many wireless networks to enhance coverage, reliability and throughput. Following \cite{Sendonaris:03} and \cite{Laneman:2004} there has been a growing interest in developing cooperative relaying techniques for wireless systems. While one-way relaying has been widely considered in the literature, in most practical communication scenarios data flows in both directions. Hence, the relay can be used to improve the performance of both transmissions simultaneously. This pragmatic approach has been modeled as the \emph{two-way relay channel} in the literature and has attracted significant recent interest \cite{Rankov:Asil:05}, \cite{Popovski:ICC:06}, \cite{Knopp:IZS:06}. Although many involved transmission schemes have been proposed for communication over two-way relay channels \cite{Rankov:Asil:05}, \cite{Gunduz:Allerton:08}, \cite{Schnurr:ISIT:08}, the capacity region remains open.

In this paper, we consider a ``separated'' two-way relay channel (sTRC) \cite{Gunduz:Allerton:08}, in which the two users can receive signals only from the relay terminal (see Fig. \ref{f:system}). In practice, this corresponds to a scenario in which the users are physically separated and the signals received from each other are negligible, such as two distant land stations communicating with a satellite, or two mobile users located on opposite sides of a building communicating with the same base station on top of the building. When there is no direct connection between the two wireless terminals, relays are essential to enable communication.

We consider multiple antennas at each terminal and model the channels between the users as quasi-static, independent, frequency non-selective Rayleigh fading channels, and assume that perfect channel state information (CSI) is available only at the receivers. Our focus here is on the diversity-multiplexing tradeoff (DMT) analysis for the multiple-input multiple-output (MIMO) sTRC. DMT analysis, introduced in \cite{Zheng:IT:03}, is useful in characterizing the fundamental tradeoff between the reliability and the number of degrees-of-freedom of a system. We measure the reliability by the diversity gain, defined as the rate of decay of the error probability with increasing $\mathrm{SNR}$, and measure the degrees-of-freedom of the system by the spatial multiplexing gain, defined as the rate of increase in the transmission rate with $\mathrm{SNR}$. The optimal DMT of a point-to-point MIMO system is characterized in \cite{Zheng:IT:03}, and it is shown to be a piecewise linear function.

\psfrag{M1}{$M_1$}
\psfrag{M2}{$M_r$}
\psfrag{M3}{$M_2$}
\psfrag{H1}{$\mathbf{H}_1$}\psfrag{H2}{$\mathbf{H}_2$}
\psfrag{H3}{$\mathbf{H}_3$}\psfrag{H4}{$\mathbf{H}_4$}
\psfrag{S}[][.3]{\huge{$\mathrm{S}_1$}}
\psfrag{R}[][.3]{\huge{$\mathrm{R}$}}
\psfrag{D}[][.3]{\huge{$\mathrm{S}_2$}}
%---------------------------
\begin{figure}
\centering
\includegraphics[width=3.5in]{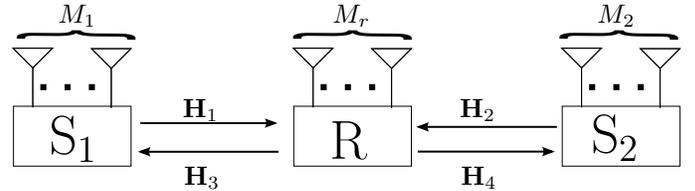}
\caption{The $(M_1, M_2, M_3)$  separated MIMO two-way relay channel model. There is no direct link between the user terminals $S1$ and $S2$.}
\label{f:system}
\end{figure}
%---------------------------

When only one of the users is active, the sTRC model reduces to a MIMO multi-hop relay channel, for which the optimal DMT is characterized in \cite{Gunduz:Globecom:08}. For a full-duplex relay terminal the optimal DMT for this multi-hop setup is achievable by decode-and-forward (DF) relaying. On the other hand, if the relay is constrained to half-duplex operation, fixed time allocation schemes fall short of the optimal DMT performance, while the dynamic decode-and-forward (DDF) protocol, introduced in \cite{Azarian:IT:05}, can be shown to be DMT-optimal.

In this paper, we show that in sTRC the DF protocol fails to achieve the optimal DMT performance even in the case of a full-duplex relay. Enforcing the relay terminal to decode both messages limits the achievable multiplexing gains due to the additional sum-rate constraint. As an alternative, we consider a coding scheme based on compress-and-forward (CF) relaying \cite{Cover:1979} in that the relay does not decode the underlying messages of the users. Instead, the relay quantizes its output and broadcasts the quantized output to the users using the coding scheme in \cite{Gunduz:Allerton:08}, \cite{Schnurr:ISIT:08}. In this coding scheme, transmission of the relay's received signal to the users is considered as a lossy joint source-channel coding problem in which the receivers (the two users in our setup) have side information (their own transmit signals) correlated with the underlying source signal. This joint source-channel coding problem is studied in \cite{Nayak:IT:08}, \cite{Gunduz:ISIT:08} in detail from the achievable distortion perspective at the receivers.

In \cite{Gunduz:Allerton:08}, it is shown that this CF scheme achieves rates within a half bit of the capacity region in the Gaussian setting with single antenna terminals. Here we show that the CF protocol achieves the optimal DMT performance in the case of multiple antenna terminals. Note that, the DMT-optimality of the CF scheme is shown in \cite{Yuksel:IT:07} for multiple antenna relay channels and in \cite{Vaze:IT:08} for multiple antenna two-way relay channels. However, both of these works assume the availability of perfect channel state information at the transmitters which is then used to decide the optimal quantization noise covariance as well as the optimal relay listening time when the relay is half-duplex. Here we propose a transmission protocol that does not necessarily choose the optimal values for these variables, and hence does not require instantaneous channel state information at the transmitters, yet achieves the optimal DMT performance.

We also consider the multi-hop MIMO relay setup (equivalent to the sTRC with only one user transmitting), and propose a dynamic CF (DCF) protocol for a half-duplex relay terminal. We show that DCF achieves the optimal DMT performance characterized in \cite{Gunduz:Globecom:08}. This idea can also be used to extend the DMT analysis of sTRC to a half-duplex relay terminal.

The rest of the paper is organized as follows. In Section \ref{s:system} we introduce the system model. Section \ref{s:FD} is devoted to the DMT analysis of sTRC with a full-duplex relay terminal. We consider both the DF and the CF protocols and show that the optimal DMT is achieved by the CF protocol. In Section \ref{s:HD}, we propose a novel dynamic CF protocol for communication over multi-hop MIMO half-duplex relay channels, and prove that it is DMT-optimal. Our conclusions and appendices containing proofs of our results follow.

\section{System Model}\label{s:system}

We consider the sTRC composed of two users $S_1$ and $S_2$, with $M_1$ and $M_2$ antennas, respectively, and a relay terminal $R$ with $M_r$ antennas, as in Fig. \ref{f:system}. We call this an $(M_1,M_r,M_2)$ system. Over a block of $T$ symbols, the received signal at the relay is
\begin{eqnarray}
\mathbf{Y}_r &=& \sqrt{\frac{\mathrm{SNR}}{M_1}} \mathbf{H}_1 \mathbf{X}_1 + \sqrt{\frac{\mathrm{SNR}}{M_2}} \mathbf{H}_2 \mathbf{X}_2 + \mathbf{W}_r,
\end{eqnarray}
while the received signal at user $i$, $i=1,2$, is
\begin{eqnarray}
\mathbf{Y}_i &=& \sqrt{\frac{\mathrm{SNR}}{M_r}} \mathbf{H}_{i+2} \mathbf{X}_r + \mathbf{W}_i.
\end{eqnarray}
Here $W_r \in \mathds{C}^{M_r \times T}$ and $W_i \in \mathds{C}^{M_i \times T}$, $i=1,2$, are the additive noise components, whose entries are complex Gaussian random variables with independent and identically distributed (i.i.d.) zero mean, variance $1/2$ Gaussian real and imaginary components, denoted by $\mathcal{CN}(0,1)$. Channels are assumed to be frequency non-selective, quasi-static Rayleigh fading and independent of each other; that is, for $i=1,2$, $\mathbf{H}_i \in \mathds{C}^{M_r \times M_i}$ and $\mathbf{H}_{i+2} \in \mathds{C}^{M_i \times M_r}$ are the channel matrices with i.i.d. entries distributed as $\mathcal{CN}(0,1)$. We have short-term power constraints at the two users and at the relay given by $\mathrm{trace}(E[\mathbf{X}_i^H \mathbf{X}_i]) \leq M_i T$ for $i=1,2,r$.

We assume that the channel realization is perfectly known by the receiving end of each transmission. We also assume that the realizations of $\mathbf{H}_1$ and $\mathbf{H}_2$ are known by the users at the end of $T$ symbols. This can be accomplished by broadcasting at a small rate $\epsilon>0$ from the relay terminal the quantized versions of both channel coefficient matrices to the users with arbitrarily small error. It can be shown that this does not effect the DMT analysis and arbitrary low quantization noise variance can be achieved as $T$ goes to infinity. %In the DMT , we assume the availability of these matrices at the users for decoding.

Following \cite{Zheng:IT:03}, for increasing $\mathrm{SNR}$ we consider a family of codes and say that the multiplexing gain of user $i$ is $r_i$ if the user's rate $R_i(\mathrm{SNR})$ satisfies, for $i=1,2$,
\[
\lim_{\mathrm{SNR} \rightarrow \infty} \frac{R_i(\mathrm{SNR})}{\log(\mathrm{SNR})} = r_i,
\]
and the corresponding diversity gain of user $i$ is $d_i$ if,
\[
d_i=-\lim_{\mathrm{SNR} \rightarrow \infty} \frac{\log P_e^i(\mathrm{SNR})}{\log(\mathrm{SNR})},
\]
in which $P_e^i(\mathrm{SNR})$ is the error probability of user $i$. In the rest of the paper, we consider codes with sufficiently long codewords so that the error event is dominated by the outage event. For each $(r_1,r_2)$ pair, we want to characterize the set of achievable diversity gain pairs $(d_1, d_2)$ denoted by $\mathcal{D}(r_1, r_2)$.

%We have characterized the DMT for a MIMO multi-hop system in \cite{Gunduz:Globecom:08} for both full-duplex (FD) and half-duplex (HD) relays. In the FD relay case, optimal DMT is achieved by the classical block Markov DF scheme. On the other hand, in the case of a half-duplex relay DDF scheme of \cite{Azarian:IT:05} achieves the optimal DMT.

\section{DMT of MIMO {s}TRC with a Full-duplex Relay}\label{s:FD}

%We first focus on the FD relay case, and denote the set of achievable multiplexing gains for a given diversity gain $d$ by $\mathcal{R}^{FD}(d)$.

We start with an outer bound on the achievable diversity gain region $\mathcal{D}(r_1, r_2)$ for any pair of multiplexing gains $(r_1,r_2)$. Then we consider DF and CF protocols and compare their achievable performances with the outer bound.

\begin{prop}\label{p:outerbound}
For any given $(r_1, r_2)$, if $(d_1, d_2) \in \mathcal{D}(r_1, r_2)$  then we have $d_i \leq d_{M^*,M_r}(r_i)$ for $i=1,2$, where  $M^*\triangleq \min\{M_1,M_2\}$ and $d_{M,N}(r)$ is the optimal DMT for a point-to-point MIMO channel with $M$ transmit and $N$ receive antennas.
\end{prop}

\begin{proof}
This outer bound can be obtained easily from the usual cut-set bounds by considering each user separately, assuming the other one is silent. For example, for user 1, the DMT is bounded by
\begin{eqnarray*}
  d_1 &\leq&  \min\{d_{M_1,M_r}(r_1), d_{M_r,M_2}(r_1)\} \\
      &=& d_{M^*,M_r}(r_1).
\end{eqnarray*}
\end{proof}

\subsection{Decode-and-Forward Relaying}\label{ss:DF}

Here, we characterize the achievable DMT with DF relaying. Due to the joint decoding requirement at the relay terminal, the DMT becomes involved when we allow different diversity gains for the two users; thus similar to \cite{Tse:IT:04} we assume that the users have the same diversity gain requirement $d$. We define $\mathcal{R}_{DF}(d)$ as the set of achievable multiplexing gain pairs for which both users achieve a diversity gain of $d$.

From Proposition \ref{p:outerbound}, it easily follows that for any $(r_1,r_2)\in \mathcal{R}_{DF}(d)$, we have $r_i \leq r_{M^*,M_r}(d)$, where $r_{M,N}(d)$ is the multiplexing-diversity tradeoff curve for a point-to-point MIMO channel with $M$ transmit and $N$ receive antennas.

In DF relaying, we use a block Markov structure, that is, each user transmits its message in blocks, and at each channel block the relay forwards the messages it decoded from the previous channel block. This transmission scheme is a combination of a multiple-access phase and a broadcast phase, which take place simultaneously for consecutive data blocks. In the broadcast phase, the relay uses the coding scheme first introduced in \cite{Tuncel:IT:06} (see also \cite{Knopp:IZS:06}, \cite{Oechtering:IT:08}) for broadcasting a source to two receivers with correlated side information. Here we consider decoded messages at the relay as the source signal, and each user's own message as the correlated side information. The achievable multiplexing gain region by DF is found as follows.

\begin{thm}\label{t:fdDF}
$\mathcal{R}^{FD}_{DF}(d) = \{(r_1, r_2): 0\leq r_i \leq r_{M^*,M_r}(d), i=1,2, \mbox{ and } r_1+r_2 \leq r_{M_1+M_2,M_r}(d) \}$.
\end{thm}

We skip the proof of the theorem due to space limitations, but the analysis of the multiple access phase follows similarly to \cite{Tse:IT:04}, while the broadcast phase follows similarly to the analysis of the CF scheme that will be given below. Note that, compared to the outer bound, the achievable multiplexing gain region in Theorem \ref{t:fdDF} has an additional constraint due to the decoding requirement at the relay. Hence, these two bounds do not meet in general.

\begin{cor}\label{c:fdDF}
For a given common diversity gain $d$, if
\begin{equation}\label{opt_DF}
r_{M^*,M_r}(d) \leq \frac{1}{2} r_{M_1+M_2,M_r}(d)
\end{equation}
then DF relaying achieves the optimal multiplexing gain region of the full-duplex MIMO sTRC.
\end{cor}

Since the optimal multiplexing gain region is a square, the optimal operating point is the corner point where both users achieve a multiplexing gain of $r_{M^*,M_r}(d)$. Also notice that, for a given system, optimality of DF is achieved for all diversity gains up to $d^*$ for which (\ref{opt_DF}) is satisfied with equality. Hence, for high diversity gains, or equivalently for low multiplexing gains, DF achieves the optimal performance for each user, in which case the transmission in either direction takes place as if the other user is silent. This is similar to the conclusion in \cite{Tse:IT:04} for multiple access channels, in which case each user achieves the optimal single-user DMT for a lightly-loaded regime of a limited number of users. Also note that DF relaying does not require any additional channel knowledge other than the CSI at the receivers.

\subsection{Compress-and-Forward Relaying}\label{ss:CF}

Now we consider a compress-and-forward (CF) transmission protocol at the relay. In this protocol, the relay broadcasts a quantized version of its received signal to the users. Since each user knows the signal it transmitted in the previous block, they have access to correlated side information. This, similar to the DF protocol, can be seen as a joint source-channel coding problem with correlated side information at the receivers. The difference is that we want lossy reconstruction at the receivers instead of lossless transmission \cite{Nayak:IT:08}, \cite{Gunduz:ISIT:08}.

In \cite{Gunduz:Allerton:08}, \cite{Schnurr:ISIT:08}, coding schemes for the sTRC based on the joint source-channel coding approach are studied. Note that the proposed coding scheme is different from the classical CF relaying in \cite{Cover:1979}, where the relay applies separate source and channel coding to its received signal. In the case of a Gaussian sTRC, it is shown in \cite{Gunduz:Allerton:08} that the proposed CF scheme achieves within a half bit of the capacity region. Here we consider this scheme with joint decoding at the receivers, that is, each receiver directly decodes the message from the received signal and its own transmitted codeword without explicitly decoding the quantized relay signal. This joint decoding scheme is proposed in \cite{Sanderovich:IT:08} for communication with decentralized processing and in \cite{Schnurr:ISIT:08} for the sTRC.

\begin{thm}\label{t:fdCF}
For given multiplexing gains $r_1$ and $r_2$, the optimal DMT for the sTRC with a full-duplex relay is characterized by
\[\mathcal{D}(r_1, r_2) = \{(d_1, d_2): 0\leq d_i\leq d_{M^*,M_r}(r_i), i=1,2\},\]
and this optimal DMT is achieved by CF relaying.
\end{thm}

\begin{proof}
A sketch of the proof of the theorem is given in Appendix \ref{a:fdCF}.
\end{proof}

Consider for example the symmetric multiplexing gain scenario with $r_1=r_2=r$. In this case, each user can achieve a diversity gain of $d_{M^*,M_r}(r)$. Consequently, compared to a system with one-way communication the CF protocol provides the same reliability while doubling the total number of degrees-of-freedom. An important question that we are currently exploring is whether this result is scalable with the number of relays or with the number of users. Note that the DF scheme can be extended to multiple relays or multiple users easily, while the proposed CF scheme requires forwarding of the CSI over the network which will increase the complexity significantly with the increasing number of relays. Furthermore, the performance of CF depends heavily on the fact that each user, knowing its own transmission, can receive a good description of the other user's signal from the quantized relay received signal. As the number of users increases this advantage will disappear due to the inevitable interference among multiple transmissions from the remaining users.

\section{Multi-hop MIMO Channel with a Half-duplex Relay}\label{s:HD}

In channels with half-duplex relays, it has been shown in \cite{Gunduz:Globecom:08} and \cite{Azarian:IT:05} that fixed time allocation between relay's listen and transmit modes is suboptimal. A dynamic DF protocol (DDF) is proposed in \cite{Azarian:IT:05} which is shown in \cite{Gunduz:Globecom:08} to achieve the optimal DMT for multi-hop MIMO relay systems. In DDF, the relay listens to the source transmission until it decodes the message, and then starts forwarding it to the destination. Our goal here is to develop a dynamic CF (DCF) protocol that achieves the optimal DMT performance in the multi-hop scenario. This is a first step towards extending our results in Section \ref{ss:CF} to sTRC with a half-duplex relay. For the analysis, we consider the same system model as in Section \ref{s:system} with $R_2=0$, i.e., only user 1 is transmitting.

In the DDF scheme, the amount of time the relay remains in the listen mode is naturally decided by the time the relay accumulates enough mutual information to be able to decode the message. Then the relay reencodes and forwards the message to the destination. However, in the CF protocol, since the relay does not decode the message, it is not clear when it should start forwarding. If, replicating DDF, the relay listens until the accumulated mutual information is exactly sufficient to decode the message, and then forwards a quantized version of its received signal to the destination, this surely will not be sufficient for decoding at the destination due to the data processing inequality.

In the DCF protocol we propose, at each channel realization, the relay listens until it could have decoded a message of rate $R_1+1$ rather than the rate of the message $R_1$. If $t(\mathbf{H}_1)$ is the random variable denoting the time the relay is in the listen mode, we have
\[
t(\mathbf{H}_1) = \frac{1 + R_1}{C_1(\mathbf{H}_1)}.
\]
The relay then quantizes its received signal up to that point, and transmits a channel codeword corresponding to this quantization codeword similarly to the coding scheme for full-duplex sTRC in Section \ref{ss:CF}. Since the relay knows the message rate and the realization of $\mathbf{H}_1$, it can decide on the listening time dynamically at each channel realization. We have the following result whose proof is sketched in Appendix \ref{a:multihop}.

\begin{thm}\label{t:multihop}
The DCF protocol is DMT-optimal for multi-hop MIMO half-duplex relay channels.
\end{thm}

%\begin{proof}
%The proof of the theorem is given in Appendix \ref{a:multihop}.
%\end{proof}
This DCF protocol, to our knowledge, is the first dynamic protocol based on CF relaying, and does not require CSI at the transmitters.

%The performance of DCF in more complicated network scenarios such as a sTRC with a half-duplex relay terminal is the topic of our current research.
%-------------------------

\section{Conclusions}

We have considered a multi-hop MIMO relay network with two-way data transmission. The capacity region for this system remains open, and other than some special cases, no known transmission strategy can simultaneously provide the users with their single user rates due to their mutual interference. We have analyzed this system in terms of the achievable diversity-multiplexing tradeoff. Interestingly, for a full-duplex relay, we have shown that the two users can simultaneously achieve their optimal DMT performance as if they were the only transmitter in the system; hence we conclude that the two users do not interfere with each other (in terms of the DMT performance).

We have shown that the optimal DMT performance is achieved by a compress-and-forward strategy that does not require the instantaneous channel state information at the transmitters. We have then proposed a dynamic version of this scheme, called the dynamic compress-and-forward (DCF) for multi-hop MIMO relay channels, and showed that this scheme achieves the optimal DMT performance. We are currently working on the performance of the DCF protocol for multi-hop MIMO two-way half-duplex relay networks.

\appendices

\section{Proof of Theorem \ref{t:fdCF}} \label{a:fdCF}

We assume Gaussian codebooks for both the channel coding and the quantization at the relay. Let $R_i=r_i \log \mathrm{SNR}$, for $i=1,2$, and define
\begin{equation}
C_i(\mathbf{H}_i) \triangleq \log \left| \mathbf{I} + \frac{\mathrm{SNR}}{M_i} \mathbf{H}_i\mathbf{H}_i^\dag \right|,
\end{equation}
for $i=1,\ldots,4$. We supress the dependence on $\mathbf{H}_i$'s for notational convenience.
%and similarly
%\begin{equation}
%C^r_i(\mathbf{H}^r_i) \triangleq \log \left| \mathbf{I} + \frac{\mathrm{SNR}}{M_i} \mathbf{H}^r_i\mathbf{H}^{r\dag}_i \right|.
%\end{equation}

At the end of each channel block, the relay first quantizes its received signal. Let
\[\hat{\mathbf{Y}}_r \triangleq \mathbf{Y}_r + \mathbf{Q},\]
where $\mathbf{Q} \in \mathds{C}^{M_r \times 1}$ is a zero-mean complex Gaussian vector with identity covariance matrix, i.e., $\mathbf{Q} \sim \mathcal{CN}(0, \mathbf{I})$. In the coding scheme, $2^{n(I(\mathbf{Y}_r;\hat{\mathbf{Y}}_r)+\epsilon)}$ quantization codewords and channel codewords are generated with distribution $P_{\hat{\mathbf{Y}}_r}$ and $P_{\mathbf{X}_r}$, respectively. For each received signal the relay finds the jointly typical quantization codeword and transmits the corresponding channel codeword. Each receiver directly decodes other user's message from its received signal and its own transmitted codeword through joint typicality decoding. See \cite{Schnurr:ISIT:08} for details of the decoding scheme and the probability of error analysis.

Consider decoding at user 2. The message of user 1 can be decoded successfully if,
\begin{align}
R_1 \leq \min \{ & [I(\mathbf{X}_r; \mathbf{Y}_2) - I(\mathbf{Y}_r; \hat{\mathbf{Y}}_r | \mathbf{H}_1,  \mathbf{H}_2, \mathbf{X}_2)]^+, \nonumber \\
& I(\mathbf{X}_1; \hat{\mathbf{Y}}_r | \mathbf{H}_1, \mathbf{H}_2, \mathbf{X}_2)  \},
\end{align}
where $[x]^+ \triangleq \max\{0,x\}$. Note here that both channel states $\mathbf{H}_1$ and $\mathbf{H}_2$ are assumed to be perfectly known at user 1 when decoding. The quantized channel states with arbitrarily small quantization noise variances can be broadcast from the relay together with the quantized received signal.

Then using the union bound, the outage probability for decoding at user 2 can be bounded as
\begin{align}\label{pout1}
 P_{out}^1 \leq \mathrm{Pr} \{R_1 > [C_4 -1]^+ \} +  \mathrm{Pr} \{R_1 > \check{C}_1 \} ,
\end{align}
where we defined, for $i=1,2$,
\begin{eqnarray}\label{d:cCapH}
\check{C}_i(\mathbf{H}_i) \triangleq \log \left|\mathbf{I} + \frac{\mathrm{SNR}}{2M_i} \mathbf{H}_i\mathbf{H}_i^\dag \right|.
\end{eqnarray}

Let $\lambda_{i,1}, \ldots, \lambda_{1,M_i^*}$ be the nonzero eigenvalues of $\mathbf{H}_i\mathbf{H}_i^\dag$, and suppose $\lambda_{i,j} =\mathrm{ SNR}^{-\alpha_{i,j}}$ for $j=1,\ldots, M_i^*$, $i=1,\ldots,4$. We have $M_3^* = M_1^* = \min\{M_1, M_r\}$ and $M_4^* = M_2^* = \min\{M_2, M_r\}$. Then we have\footnote{Define the exponential equality as $f(\mathrm{SNR}) \dot{=} \mathrm{SNR}^c$, if $\lim_{\mathrm{SNR} \rightarrow \infty} \frac{\log f(\mathrm{SNR})}{\log \mathrm{SNR}} = c$. The exponential inequalities $\dot{\leq}$ and $\dot{\geq}$ are defined similarly.}

\begin{eqnarray}\label{e:DDFout1}
C_i(\mathbf{H}_i) &= & \log \prod_{j=1}^{M_i^*} \left(1+\frac{\mathrm{SNR}}{M_i} \lambda_{i,j} \right) \nonumber \\
 &\doteq & \log \prod_{j=1}^{M_i^*} \mathrm{SNR}^{(1-\alpha_{i,j})^+}. \label{e:DDFout2}
\end{eqnarray}
We also have $\check{C}_i(\mathbf{H}_i) \doteq C_i(\mathbf{H}_i)$. Using these exponential equalities, we can obtain the following from (\ref{pout1}).
\begin{align}
P_{out}^1 \doteq & \mathrm{Pr} \left\{ r_1 \log \mathrm{SNR} > \log \mathrm{SNR}^{S_4(\boldsymbol{\alpha_4})} \right\} \nonumber \\
& + \mathrm{Pr} \left\{ r_1 \log \mathrm{SNR} > \log \mathrm{SNR}^{S_1(\boldsymbol{\alpha_1})} \right\}  \nonumber \\
\doteq & \mathrm{SNR}^{-d_{M_r, M_2} (r_1)} + \mathrm{SNR}^{-d_{M_1, M_r} (r_1)} \nonumber\\
\doteq & \mathrm{SNR}^{-d_{M^*, M_r} (r_1)} \nonumber
\end{align}
where we have defined $S_i(\boldsymbol{\alpha_i}) \triangleq \sum_{j=1}^{M_i^*} (1-\alpha_{i,j})^+$ for $i=1,\ldots,4$, and used the result from \cite{Zheng:IT:03} for the second exponential equality. The same arguments can be similarly applied to the other message as well. Comparing this achievable DMT of the CF protocol with the outer bound in Proposition \ref{p:outerbound} concludes the proof.

%hence leading to the desired result of
%\[d_i(r_i) = d_{M^*,M_r}(r_i).\]

\section{Proof of Proposition \ref{t:multihop} \label{a:multihop}}
Due to space limitations, we give only a sketch of the proof here. The relay listens for $tT$ symbols with $t = \frac{R_1 + 1}{C_1}$. We have outage if $t\geq 1$. In the remaining $\bar{t}T$ symbols, with $\bar{t}\triangleq 1-t$, the relay quantizes the received symbols and transmits the corresponding channel codeword as in Appendix \ref{a:fdCF}.

The probability of outage can be bounded as follows.
\begin{align}
P_{out}^1 \leq & \mathrm{Pr} \left\{t > 1 \right\} + \mathrm{Pr} \left\{R_1> \bar{t} C_4 - t \right\}  \nonumber \\
& + \mathrm{Pr} \left\{R_1> t \check{C}_1  \right\} \\
     = &  \mathrm{Pr} \left\{R_1> \left(1-\frac{1+R_1}{C_1}\right)C_4 - \frac{1+R_1}{C_1} \right\}   \nonumber \\
    & + \mathrm{Pr} \left\{R_1 > C_1 -1 \right\} +  \mathrm{Pr} \left\{R_1> \frac{1+R_1}{C_1} \check{C}_1  \right\} \\
    = & \mathrm{Pr} \left\{R_1> \frac{C_1C_4-C_4-1}{1+C_1+C_4} \right\}+ \mathrm{Pr} \left\{R_1 > C_1-1 \right\} \nonumber \\
    & +  \mathrm{Pr} \left\{R_1 > \frac{\check{C}_1}{C_1-\check{C}_1}   \right\} \nonumber \\
    \dot{\leq} & \mathrm{Pr} \left\{R_1> \frac{C_1C_4-C_4-1}{1+C_1+C_4} \right\} +  \mathrm{Pr} \left\{R_1 > C_1 -1   \right\} \nonumber
\end{align}
where the last inequality follows since $C_1 > \check{C}_1 \geq C_1-1$.

%On the other hand we also have
%\begin{align}
%P\{t \geq 1\} & =  P \left\{ 1 + r\log \mathrm{SNR} > C_1 \right\}.  \label{pout2}
%\end{align}
%However, we can ignore this outage event since (\ref{pout2}) is always included in the outage event in (\ref{pout1}).

%The probability that the destination can not decode the message from the quantized relay signal is found as
%\begin{align}
%P& \left\{\frac{t}{2} \check{C}_1 < r \log \mathrm{SNR} \right\} \nonumber \\
%& = P \left\{ r\log \mathrm{SNR} \left(C_1- \frac{\check{C}_1}{2} \right) > \frac{\check{C}_1}{2} \right\} \nonumber \\
%& \leq  P \left\{ r\log \mathrm{SNR} > \frac{\check{C}_1}{2} \right\},  \nonumber \\
%& \leq  P \left\{ r\log \mathrm{SNR} > C_1 -1  \right\},  \label{pout3}
%\end{align}
%where the last two inequalities follow since $C_1 - 1 \leq \check{C}_1/2$. Similar to (\ref{pout2}), we can ignore this event as well.

Then we have
\begin{align}
P_{out}^1 \dot{\leq} & \mathrm{Pr} \left\{R_1 > \frac{C_1 C_4}{C_1+C_4} \right\} + \mathrm{Pr} \left\{R_1 > C_1 \right\} \nonumber\\
    \doteq &  \mathrm{Pr} \left\{R_1 > \frac{C_1 C_4}{C_1+C_4} \right\}  \nonumber\\
    \doteq & P\left\{r_1 \log \mathrm{SNR} > \frac{\log \mathrm{SNR}^{S_1(\boldsymbol{\alpha_1})} \log \mathrm{SNR}^{S_4(\boldsymbol{\alpha_4})}}{\log \mathrm{SNR}^{S_1(\boldsymbol{\alpha_1})} + \log \mathrm{SNR}^{S_4(\boldsymbol{\alpha_4})}} \right\} \nonumber
\end{align}
\begin{align}
\doteq & P\left\{r_1 > \frac{S_1(\boldsymbol{\alpha_1}) S_4(\boldsymbol{\alpha_4})} {S_1(\boldsymbol{\alpha_1}) + S_4(\boldsymbol{\alpha_4})} \right\} \nonumber
\end{align}
where, as before, we have $S_i(\boldsymbol{\alpha_i}) \triangleq \sum_{j=1}^{M_i^*} (1-\alpha_{i,j})^+$.

Now, comparison of this outage probability expression with the outer bound given in \cite{Gunduz:Globecom:08} reveals that the two match, and we achieve the optimal DMT performance. We do not have a general explicit expression for the optimal DMT of a MIMO multi-hop relay system as in the case of point-to-point MIMO channels \cite{Zheng:IT:03}; however the tradeoff can be obtained numerically from the optimization problem given in \cite{Gunduz:Globecom:08}.

\bibliographystyle{ieeetran}
\bibliography{refasil}

\begin{thebibliography}{10}
\providecommand{\url}[1]{#1}
\csname url@samestyle\endcsname
\providecommand{\newblock}{\relax}
\providecommand{\bibinfo}[2]{#2}
\providecommand{\BIBentrySTDinterwordspacing}{\spaceskip=0pt\relax}
\providecommand{\BIBentryALTinterwordstretchfactor}{4}
\providecommand{\BIBentryALTinterwordspacing}{\spaceskip=\fontdimen2\font plus
\BIBentryALTinterwordstretchfactor\fontdimen3\font minus
  \fontdimen4\font\relax}
\providecommand{\BIBforeignlanguage}[2]{{%
\expandafter\ifx\csname l@#1\endcsname\relax
\typeout{** WARNING: IEEEtran.bst: No hyphenation pattern has been}%
\typeout{** loaded for the language `#1'. Using the pattern for}%
\typeout{** the default language instead.}%
\else
\language=\csname l@#1\endcsname
\fi
#2}}
\providecommand{\BIBdecl}{\relax}
\BIBdecl

\bibitem{Sendonaris:03}
A.~Sendonaris, E.~Erkip, and B.~Aazhang, ``User cooperation diversity--part
  {I}: System description,'' \emph{IEEE Trans. on Communications}, vol.~51,
  no.~11, pp. 1927--1938, November 2003.

\bibitem{Laneman:2004}
J.~N. Laneman, D.~N.~C. Tse, and G.~W. Wornell, ``Cooperative diversity in
  wireless networks: Efficient protocols and outage behavior,'' \emph{IEEE
  Trans. on Information Theory}, vol.~50, no.~12, pp. 3062--3080, December
  2004.

\bibitem{Rankov:Asil:05}
B.~Rankov and A.~Wittneben, ``Spectral efficient signaling for half-duplex
  relay channels,'' in \emph{Proc. 39th Asilomar Conference on Signals,
  Systems, and Computers}, Pacific Grove, CA, November 2005.

\bibitem{Popovski:ICC:06}
P.~Popovski and H.~Yomo, ``The anti-packets can increase the achievable
  throughput of a wireless multi-hop network,'' in \emph{Proc. IEEE Int'l Conf.
  on Communication (ICC)}, Istanbul, Turkey, January.

\bibitem{Knopp:IZS:06}
R.~Knopp, ``Two-way radio networks with a star topology,'' in \emph{Proc. IEEE
  Int'l Zurich Seminar on Communications}, Zurich, Switzerland, February 2006.

\bibitem{Gunduz:Allerton:08}
D.~G\"{u}nd\"{u}z, E.~Tuncel, and J.~Nayak, ``Rate regions for the separated
  two-way relay channel,'' in \emph{Proc. 46th Annual Allerton Conf. on Comm.,
  Control, and Computing}, Monticello, IL, September 2008.

\bibitem{Schnurr:ISIT:08}
C.~Schnurr, S.~Stanczak, and T.~J. Oechtering, ``Coding theorems for the
  restricted half-duplex two-way relay channel with joint decoding,'' in
  \emph{Proc. IEEE Int'l Symposium on Information Theory}, Toronto, Canada,
  July 2008.

\bibitem{Zheng:IT:03}
L.~Zheng and D.~Tse, ``Diversity and multiplexing: A fundamental tradeoff in
  multiple antenna channels,'' \emph{IEEE Trans. on Information Theory},
  vol.~49, no.~5, pp. 1073--1096, May 2003.

\bibitem{Gunduz:Globecom:08}
D.~G\"{u}nd\"{u}z, A.~Goldsmith, and H.~V. Poor, ``Diversity-multiplexing
  tradeoffs in {MIMO} relay channels,'' in \emph{Proc. IEEE Global
  Communications Conf. (Globecom)}, New Orleans, LA, November 2003.

\bibitem{Azarian:IT:05}
K.~Azarian, H.~{El~Gamal}, and P.~Schniter, ``On the achievable
  diversity-multiplexing tradeoff in half-duplex cooperative channels,''
  \emph{IEEE Trans. on Information Theory}, vol.~51, no.~12, pp. 4152--4172,
  December 2005.

\bibitem{Cover:1979}
T.~M. Cover and A.~El~Gamal, ``Capacity theorems for the relay channel,''
  \emph{IEEE Trans. on Information Theory}, vol.~25, no.~5, pp. 572 -- 584,
  September 1979.

\bibitem{Nayak:IT:08}
J.~Nayak, E.~Tuncel, and D.~G\"{u}nd\"{u}z, ``Wyner-{Z}iv coding over broadcast
  channels: {D}igital schemes,'' \emph{IEEE Trans. on Information Theory},
  submitted, 2008.

\bibitem{Gunduz:ISIT:08}
D.~G\"{u}nd\"{u}z, J.~Nayak, and E.~Tuncel, ``Wyner-{Z}iv coding over broadcast
  channels using hybrid digital/analog transmission,'' in \emph{Proc. IEEE
  Int'l Symposium on Information Theory}, Toronto, Canada, July 2008.

\bibitem{Yuksel:IT:07}
M.~Yuksel and E.~Erkip, ``Multi-antenna cooperative wireless systems: A
  diversity multiplexing tradeoff perspective,'' \emph{IEEE Trans. on
  Information Theory}, vol.~53, no.~10, pp. 3371--3393, October 2007.

\bibitem{Vaze:IT:08}
R.~Vaze and R.~W. Heath, ``On the capacity and diversity-multiplexing tradeoff
  of the two-way relay channel,'' submitted, October
  2008~[http://arxiv.org/abs/0810.3900].

\bibitem{Tse:IT:04}
D.~Tse, P.~Viswanath, and L.~Zheng, ``Diversity-multiplexing tradeoff in
  multiple access channels,'' \emph{IEEE Trans. on Information Theory},
  vol.~50, no.~9, pp. 1859--1874, September 2004.

\bibitem{Tuncel:IT:06}
E.~Tuncel, ``Slepian-{W}olf coding over broadcast channels,'' \emph{IEEE Trans.
  on Information Theory}, vol.~52, no.~4, pp. 1469--1482, April 2006.

\bibitem{Oechtering:IT:08}
T.~J. Oechtering, C.~Schnurr, I.~Bjelakovic, and H.~Boche, ``Broadcast capacity
  region of two-phase bidirectional relaying,'' \emph{IEEE Trans. on
  Information Theory}, vol.~54, no.~1, pp. 454--458, January 2008.

\bibitem{Sanderovich:IT:08}
A.~Sanderovich, S.~Shamai, Y.~Steinberg, and G.~Kramer, ``Communication via
  decentralized processing,'' \emph{IEEE Trans. on Information Theory},
  vol.~54, no.~7, pp. 3008--3023, July 2008.

\end{thebibliography}

\end{document}